\documentclass[11pt]{amsart}

\usepackage{amssymb}
\usepackage[a4paper]{geometry}
\usepackage[parfill]{parskip} 
\usepackage{color}
\usepackage{url}
\usepackage{graphicx}
\usepackage{url}
\usepackage{float}
\newcommand{\leo}{\textcolor{black}}
\newcommand{\leoo}{\textcolor{black}}
\newcommand{\leooo}{\textcolor{black}}
\newcommand{\cel}{\textcolor{black}}
\newcommand{\celB}{\textcolor{black}}

\newcommand{\low}{L_{ow}}

\newcommand{\cL}{\mathcal{L}}
\newcommand{\cR}{\mathcal{R}}
\newcommand{\cM}{\mathcal{M}}
\newcommand{\cN}{\mathcal{N}}

\newcommand{\cT}{\mathbb{T}}
\newcommand{\dB}{\mathbb{B}}

\newcommand{\cK}{\mathcal{K}}
\newcommand{\cW}{\mathcal{W}}
\newcommand{\cO}{\mathcal{O}}
\newcommand{\cC}{\mathcal{C}}

\newcommand{\cD}{\mathcal{D}}

\newtheorem{definition}{Definition}
\newtheorem{theorem}{Theorem}
\newtheorem{lemma}{Lemma}
\newtheorem{observation}{Observation}
\newtheorem{corollary}{Corollary}

\title{Reconstructing phylogenetic level-1 networks from nondense binet and trinet sets}
\author{Katharina Huber}
  \address[Katharina Huber]{School of Computing Sciences,  University of East Anglia, Norwich, United Kingdom}
\author{Leo van Iersel}
  \address[Leo van Iersel]{Delft Institute of Applied Mathematics, Delft University of Technology, The Netherlands}
\author{Vincent Moulton}
  \address[Vincent Moulton]{School of Computing Sciences,  University of East Anglia, Norwich, United Kingdom}
\author{Celine Scornavacca}
  \address[Celine Scornavacca]{ISEM, CNRS -- Universit\'e Montpellier II,  Montpellier, France\\  Institut de Biologie Computationnelle, Montpellier, France}
\author{Taoyang Wu}
  \address[Taoyang Wu]{School of Computing Sciences,  University of East Anglia, Norwich, United Kingdom}
  
\email[Leo van Iersel]{l.j.j.v.iersel@gmail.com}

\begin{document}

\begin{abstract} \emph{Binets} and \emph{trinets} are phylogenetic networks with two and three leaves, respectively. Here we consider the problem of deciding if there exists a binary level-1 phylogenetic network displaying a given set~$\cT$ of binary binets or trinets over a taxa set~$X$, and constructing such a network whenever it exists. We show that this is NP-hard for trinets but polynomial-time solvable for binets. Moreover, we show that the problem is still polynomial-time solvable for inputs consisting of binets and trinets as long as the cycles in the trinets have size three. Finally, we present an~$O(3^{|X|} poly(|X|))$ time algorithm for general sets of binets and trinets. The latter two algorithms generalise to instances containing level-1 networks with arbitrarily many leaves, and thus provide some of the first supernetwork algorithms for computing networks from a set of rooted phylogenetic networks.
\end{abstract}

\maketitle

\pagestyle{myheadings}
\thispagestyle{plain}
\markboth{HUBER, VAN IERSEL, MOULTON, SCORNAVACCA AND WU}{RECONSTRUCTING LEVEL-1 NETWORKS FROM NONDENSE BINETS OR TRINETS}

\section{Introduction}

A key problem in biology is to reconstruct
the evolutionary history of a set of taxa using data
such as DNA sequences or morphological features. These histories
are commonly represented by phylogenetic trees, and 
can be used, for example, to inform genomics studies, analyse
virus epidemics and understand the origins of humans \cite{SempleSteel2003}. 
Even so, in case evolutionary processes such as recombination and hybridization 
are involved, it can be more appropriate to
represent histories using phylogenetic networks instead 
of trees~\cite{expanding}. 

Generally speaking, a phylogenetic network is any type of graph
with a subset of its vertices labelled by the set of taxa 
that in some way represents evolutionary relationships 
between the taxa \cite{HusonRupp2008}. Here, 
however, we focus on a special type of 
phylogenetic network called a {\em level-1} network.
We present the formal definition for this type of network in 
the next section but, essentially, it is
a binary, directed acyclic graph with a single root,
whose leaves correspond to the taxa, and in 
which any two cycles are disjoint (for example, see 
Figure~\ref{fig:network} below).
This type of network was first considered in \cite{wang}
and is closely related to so-called
galled-trees \cite{Cardona2009a,GusfieldEtAl2004}.Level-1 networks have been used to, for example, analyse virus evolution \cite{lev1athan}, and are of practical importance since their simple structure allows for 
efficient construction \cite{GusfieldEtAl2004,lev1athan,huynh}
and comparison \cite{JL2014}.

One of the main approaches that has been
used to construct level-1 networks
is from {\em triplet} sets, that is, sets of rooted binary trees 
with three leaves (see e.g. \cite{lev1athan,JanssonSung2006,JanssonEtAl2006}). 
Even so, it has been 
observed that the set of triplets displayed by a level-1 network
does not necessarily provide all of the information 
required to uniquely define or {\em encode} the 
network \cite{GambetteHuber2012}. Motivated by this
observation, in \cite{huber2011encoding} an algorithm 
was developed for constructing level-1 networks from a network
analogue of triplets: rooted binary networks
with three leaves, or {\em trinets}. This algorithm relies on 
the fact that the trinets displayed by a level-1 network do indeed 
encode the network \cite{huber2011encoding}. 
Even so, the algorithm 
was developed for {\em dense} trinet sets only, i.e. sets in which  
there is a trinet associated to each combination of three taxa.

In this paper, we consider the problem of
constructing level-1 networks from arbitrary sets of level-1 trinets
and binets, where a binet is an even simpler building
block than a trinet, consisting of a rooted binary 
network with just two leaves. We consider binets
as well \leoo{as trinets} since they can provide important information to help  
piece together sets of trinets. Our approach can be regarded as a generalisation of the well-known supertree algorithm called \textsc{Build} \cite{AhoEtAl1981,SempleSteel2003} for checking whether or not a set of triplets is displayed by a phylogenetic tree and
constructing such a tree if this is the case.
In particular, the algorithm we present in Section 4 is one of the first supernetwork algorithms for constructing
a phylogenetic network from a set of networks. Note that
some algorithms have already been developed for computing {\em unrooted} supernetworks -- see, for example~\cite{HollandConner,supernetworks}.

\leoo{We now summarize the contents of the paper. After introducing some basic notation in the next section, in Section~\ref{sec:binets} we begin by presenting a polynomial-time algorithm for deciding whether or not there exists some level-1 network that displays a given set of level-1 binets, and for constructing such a network if it actually exists (see Theorem~\ref{thm:binets}). Then, in Section~\ref{sec:nondense}, we present an exponential-time algorithm for an arbitrary set consisting of binets and trinets (see Theorem~\ref{thm:exptime}). This algorithm uses a top-down approach that is somewhat similar in nature to the \textsc{Build} algorithm \cite{AhoEtAl1981,SempleSteel2003} but it is considerably more intricate.} The algorithm can be generalised to instances containing level-1 networks with arbitrarily many leaves since trinets encode level-1 networks~\cite{huber2011encoding}.

In Section~\ref{sec:tinycycle}
we show that for the special instance 
where each cycle in the input trinets has size three
our exponential-time algorithm is actually
guaranteed to work in polynomial time. This is still the case when the input consists of binary level-1 networks with arbitrarily many leaves as long as all their cycles have length three. However, in Section~\ref{sec:hardness} 
we prove that in general it is NP-hard to decide 
whether or not there 
exists a binary level-1 network that displays
an arbitrary set of trinets (see Theorem~\ref{nphard}). We also show that this problem remains NP-hard if we insist that the network contains only one cycle.
Our proof is similar to the proof 
that it is NP-hard to decide the same question
for an arbitrary set of triplets given in \cite{JanssonSung2006}, but 
the reduction is more complicated.
In Section~\ref{sec:concl}, we conclude with a discussion of some directions for future work.

\section{Preliminaries\label{sec:prelim}}

Let~$X$ be some finite set of labels. We will refer to the elements of~$X$ as \emph{taxa}. A \emph{rooted phylogenetic network} $N$ on~$X$ is a directed acyclic graph which has a single indegree-0 vertex (the \emph{root}, denoted by $\rho(N)$), no indegree-1 outdegree-1 vertices and its outdegree-0 vertices (the \emph{leaves}) bijectively labelled by the elements of~$X$. We will refer to rooted phylogenetic networks as \emph{networks} for short. In addition, we will identify each leaf with its label and denote the set of leaves of a network~$N$ by~$\cL(N)$. For a set~$\cN$ of networks, $\cL(\cN)$ is defined to be $\cup_{N\in\cN} \cL(N)$. A network is called \emph{binary} if all vertices have indegree and outdegree at most two and all vertices with indegree two (the \emph{reticulations}) have outdegree one. A \emph{cycle} of a network is the union of two non-identical, internally-vertex-disjoint, directed $s$-$t$ paths, for any two distinct vertices~$s,t$. The \emph{size} of the cycle is the number of vertices that are on at least one of these paths. A cycle is \emph{tiny} if it has size three and \emph{largish} otherwise. A network is said to be a \emph{tiny cycle network} if all its cycles are tiny. A binary network is said to be a binary \emph{level-1} network if any two cycles are disjoint. We only consider binary level-1 networks in this paper, see Figure~\ref{fig:network} for an example containing one tiny and two largish cycles.

If~$N$ is a network on~$X$ and~$X'\subseteq X$ nonempty, then a vertex~$v$ of~$N$ is a \emph{stable ancestor} of~$X'$ (in~$N$) if every directed path from the root of~$N$ to a leaf in~$X'$ contains~$v$. The \emph{lowest stable ancestor} of~$X'$ is the unique vertex~$LSA(X')$ that is a stable ancestor of~$X'$ and such that there is no directed path from $LSA(X')$ to any of the other stable ancestors of~$X'$.

A \emph{binet} is a network with exactly two leaves and a \emph{trinet} is a network with exactly three leaves. In this paper, we only consider binary level-1 binets and trinets. There exist two binary level-1 binets and eight binary level-1 trinets \leooo{(up to relabelling)}, all presented in Figure~\ref{fig:nets}. In the following, we will use the names of the trinets and binets indicated in that figure. For example, $T_1(x,y;z)$ denotes the only rooted tree on $\{x,y,z\}$ with $\{x,y\}$ as cluster.

A set~$\dB$ of binets on a set~$X$ of taxa is called \emph{dense} if for each pair of taxa from~$X$ there is at least one binet in~$\dB$ on those taxa. A set~$\cT$ of (binets and) trinets on~$X$ is \emph{dense} if for each combination of three taxa from~$X$ there is at least one trinet in~$\cT$ on those taxa.

\begin{figure}
  \centerline{\includegraphics{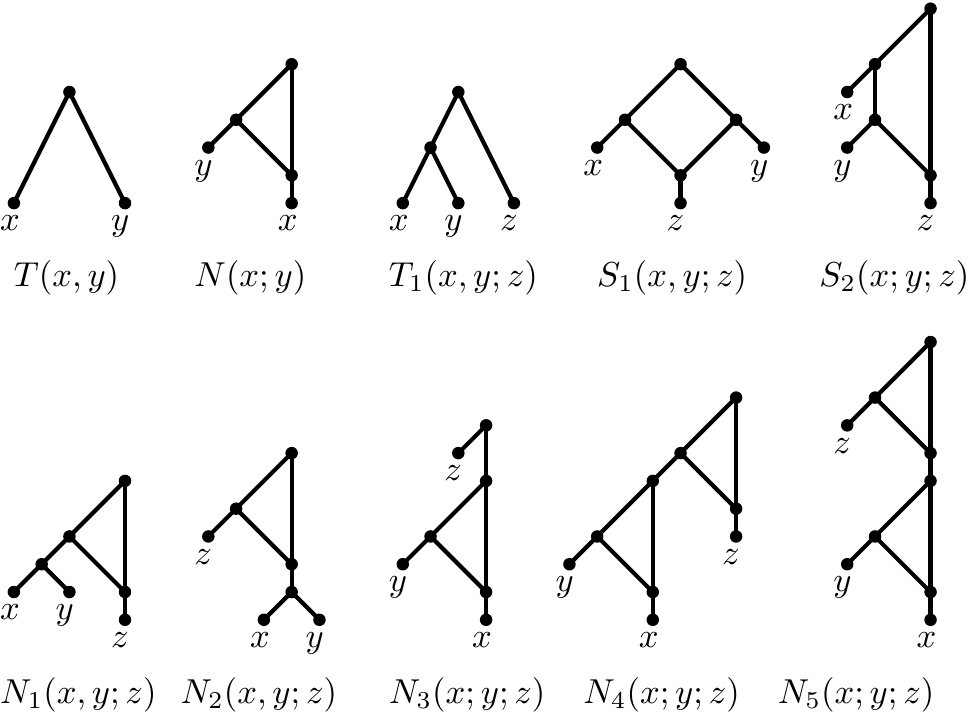}}
  \caption{The two binary level-1 binets and the eight binary level-1 trinets.}
  \label{fig:nets}
\end{figure}

Given a phylogenetic network~$N$ on~$X$ and a subset~$X'\subseteq X$, we define the network~$N|X'$ as the network obtained from~$N$ by deleting all vertices and arcs that are not on a directed path from the lowest stable ancestor of~$X'$ to a leaf in~$X'$ and repeatedly suppressing indegree-1 outdegree-1 vertices and replacing parallel arcs by single arcs until neither operation is applicable.

Two networks~$N,N'$ on~$X$ are said to be \emph{equivalent} if there exists an isomorphism between~$N$ and~$N'$ that maps each leaf of~$N$ to the leaf of~$N'$ with the same label.

Given two networks~$N,N'$ with~$\cL(N')\subseteq \cL(N)$, we say that~$N$ \emph{displays}~$N'$ if~$N|\cL(N')$ is equivalent to~$N'$. Note that this definition in particular applies to the cases that~$N'$ is a binet or trinet. In addition, we say that~$N$ \emph{displays} a set~$\cN$ of networks if~$N$ displays each network in~$\cN$.

Given a network~$N$, we use the notation~$\cT(N)$ to denote the set of all trinets and binets displayed by~$N$. For a set~$\cN$ of networks, $\cT(\cN)$ denotes $\bigcup_{N\in\cN} \cT(N)$. Given a set~$\cT$ of trinets and/or binets on~$X$ and a nonempty subset~$X'\subseteq X$, we define the restriction of~$\cT$ to~$X'$ as
\[ 
\cT|X':= \{T|(\cL(T)\cap X') \quad : \quad T\in\cT \text{ and } |\cL(T)\cap X'| \in \{2,3\} \}.
\]

The following observation will be useful.

\begin{observation}\label{obs:subset}
Let~$\cT$ be  a set of trinets and binets on~$X$  and suppose that there exists  a binary level-1 network $N$ on~$X$ such that~$\cT\subseteq\cT(N)$. Then, for any nonempty subset~$X'$ of~$X$,~$N|X'$ is a binary level-1 network displaying~$\cT|X'$.
\end{observation}

\begin{figure}
  \centerline{\includegraphics{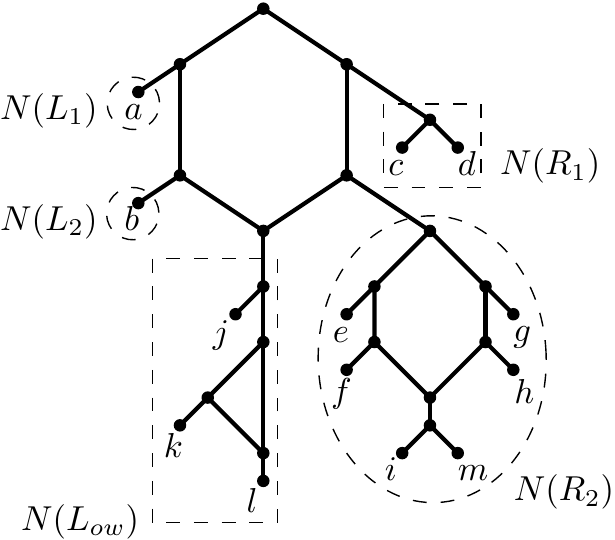}}
  \caption{A binary level-1 network~$N$. Its set of high leaves is~$H=\{a,b,c,d,e,f,g,h,i,m\}$. The bipartition of~$H$ induced by~$N$ is~$\{\{a,b\},\{c,d,e,f,g,h,i,m\}\}$. Hence, $a$ and~$b$ are on the same side in~$N$ and~$e,f,g,h,i$ and~$m$ are on the same side in~$N$. The pendant sidenetworks of~$N$ are~$N(L_1)$, $N(L_2)$, $N(R_1)$ and $N(R_2)$. Leaves~$j,k$ and~$l$ are low in~$N$.}
  \label{fig:network}
\end{figure}

\leooo{We call a network \emph{cycle-rooted} if its root is contained in a cycle. A cycle-rooted network is called  \emph{tiny-cycle rooted} if its root is in a tiny cycle and \emph{largish-cycle rooted} otherwise.} If~$N$ is a cycle-rooted binary level-1 network whose root~$\rho(N)$ is in cycle~$C$, then there exists a unique reticulation~$r$ that is contained in~$C$. We say that a leaf~$x$ is \emph{high} in~$N$ if there exists a path from~$\rho(N)$ to~$x$ that does not pass through~$r$, otherwise we say that~$x$ is \emph{low} in~$N$. If~$N$ is not cycle-rooted, then we define all leaves to be high in~$N$. We say that two leaves are \emph{at the same height} in~$N$ if they are either both high or both low in~$N$. Two leaves~$x,y$ that are both high in~$N$ are said to be \emph{on the same side} in~$N$ if~$x$ and~$y$ are both reachable from the same child of~$\rho(N)$ by two directed paths. A bipartition~$\{L,R\}$ of the set~$H$ of leaves high in~$N$ is the bipartition of~$H$ \emph{induced} by~$N$ if all leaves in~$L$ are on the same side $S_L$ in~$N$ and all the leaves in~$R$ are on the same side $S_R$ in~$N$, with $S_L\neq S_R$. 
Finally, if~$N$ is cycle-rooted, we say that a subnetwork~$N'$ of~$N$ is a \emph{pendant subnetwork} if there exists some cut-arc~$(u,\rho')$ in~$N$ with~$\rho'$ the root of~$N'$ and~$u$ a vertex of the cycle containing the root of~$N$. If, in addition,~$u$ is not a reticulation, then~$N'$ is said to be a \emph{pendant sidenetwork} of~$N$. If, in addition,~$\cL(N')\subseteq S$ with~$S$ a part of the bipartition of the high leaves of~$N$ induced by~$N$, then we say that~$N'$ is a \emph{pendant sidenetwork on side}~$S$. See Figure~\ref{fig:network} for an illustration of these definitions.

\section{Constructing a network from a set of binets}\label{sec:binets}

In this section we describe a polynomial-time algorithm for deciding if there exists some binary level-1 network displaying a given set~$\dB$ of binets, and constructing such a network if it exists. We treat this case separately because it is much simpler than the trinet algorithms and gives an introduction to the techniques we use.

The first step of the algorithm is to construct the graph $\cR^b(\dB)$\footnote{The superscript~$b$ indicates that this definition is only used for binets. In Section~\ref{sec:nondense}, we will introduce a graph $\cR(\cT)$ which will be used for general sets of binets and trinets and is a generalisation of  $\cR^b(\dB)$ in the sense that  $\cR^b(\dB) =  \cR(\dB)$ if~$\dB$ contains only binets.}, which has a vertex for each taxon and an edge~$\{x,y\}$ if (at least) one of $N(x;y)$ and $N(y;x)$ is contained in~$\dB$.

If the graph $\cR^b(\dB)$ is disconnected and has connected components~$X_1,\ldots,X_p$, then the algorithm constructs a network~$N$ by recursively computing networks $N|X_1,\ldots ,N|X_p$ displaying $\dB|X_1,\ldots ,\dB|X_p$ respectively, creating a new root node~$\rho$ and adding arcs from~$\rho$ to the roots of~$N|X_1,\ldots ,N|X_p$, and refining arbitrarily the root~$\rho$ in order to make the network binary. See Figure~\ref{fig:binetsex} for an example.

\begin{figure}
  \centerline{\includegraphics{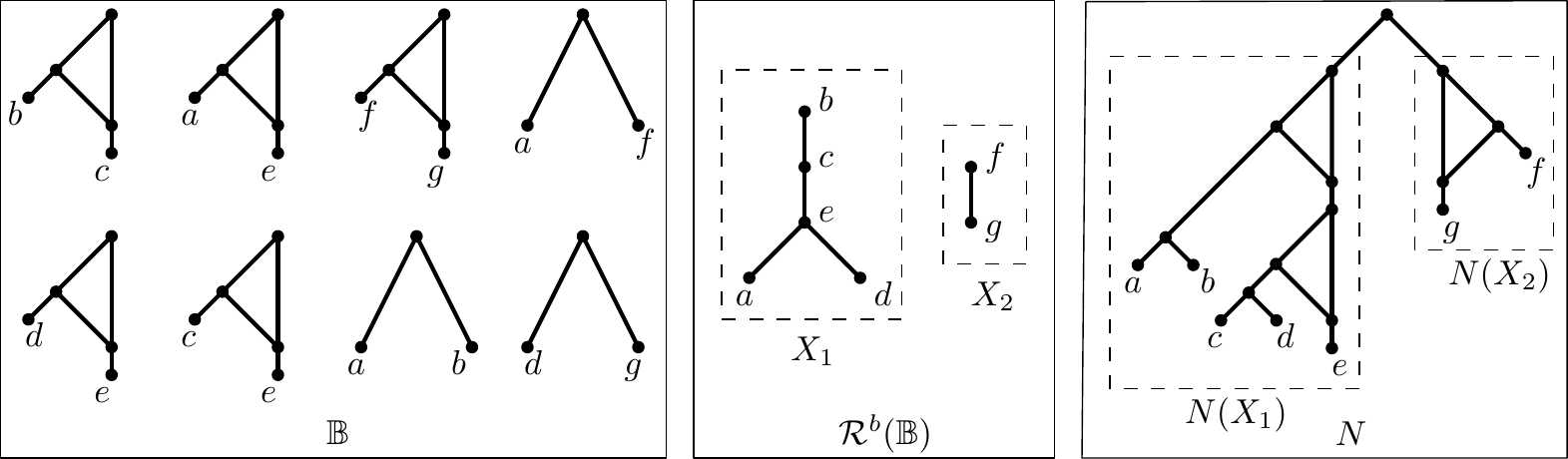}}
  \caption{Example of a step of the algorithm for constructing a network~$N$ from the set~$\dB$ of binets. The graph~$\cR^b(\dB)$ has connected components $X_1=\{a,b,c,d,e\}$ and~$X_2=\{f,g\}$. Hence, network~$N$ is obtained by combining recursively computed networks~$N(X_1)$ and~$N(X_2)$ by hanging them below a new root. See Figure~\ref{fig:binetsex2} for the first recursive step.}
  \label{fig:binetsex}
\end{figure}

If the graph $\cR^b(\dB)$ is connected, then the algorithm constructs the graph~$\cK^b(\dB)$, which has a vertex for each taxon and an edge~$\{x,y\}$ precisely if~$T(x,y)\in\dB$. In addition, the algorithm constructs the directed graph $\Omega^b(\dB)$, which has a vertex for each connected component of $\cK^b(\dB)$ and an arc $(\pi_1,\pi_2)$ precisely if there exists a binet $N(y;x)\in\dB$ with $x\in V(\pi_1)$ and $y\in V(\pi_2)$ (with~$V(\pi)$ denoting the vertex set of a given connected component~$\pi$).

The algorithm searches for a nonempty strict subset~$U$ of the vertices of $\Omega^b(\dB)$ such that there is no arc~$(\pi_1,\pi_2)$ with $\pi_1\notin U$ and $\pi_2\in U$. If there exists no such set~$U$ then the algorithm halts and outputs that there exists no solution. Otherwise, let~$H$ be the union of the vertex sets of the connected components of $\cK^b(\dB)$ that correspond to elements of~$U$ and define~$\low=X\setminus H$. The algorithm recursively constructs networks $N(H)$ displaying $\dB|H$ and $N(\low)$ displaying $\dB|\low$. Subsequently, the algorithm constructs a network~$N$ consisting of vertices~$\rho,v,r$, arcs~$(\rho,v)$, $(v,r)$, $(\rho,r)$, networks~$N(\low), N(H)$ and an arc from~$v$ to the root of~$N(H)$ and an arc from~$r$ to the root of~$N(\low)$. See Figure~\ref{fig:binetsex2} for an example of this case.

\begin{figure}
  \centerline{\includegraphics{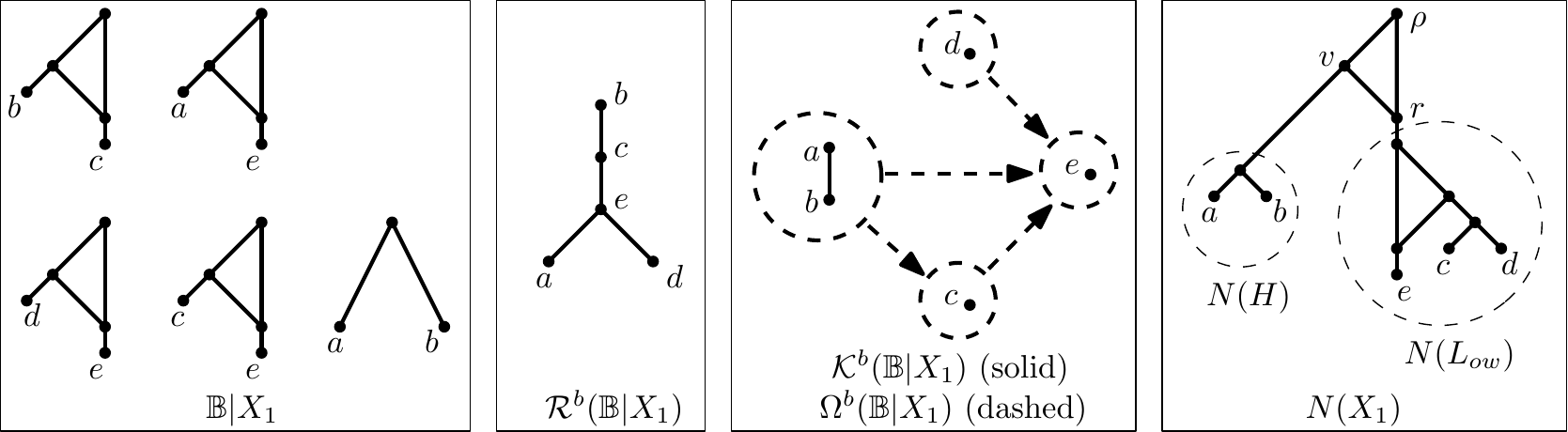}}
  \caption{Example of the recursive step, which constructs a network~$N(X_1)$ from binet set~$\dB|X_1$, with~$\dB$ and~$X_1$ as in Figure~\ref{fig:binetsex}. Network $N(X_1)$ is cycle-rooted because graph~$\cR^b$ is connected. One possible strict subset of the vertices of $\Omega^b(\dB|X_1)$ with no incoming arcs is~$\{a,b\}$. Hence, $H=\{a,b\}$ can be made the high leaves of the network, and $\low=\{c,d,e\}$ the low leaves. Combining recursively computed networks $N(H)$ and $N(\low)$ by hanging them below a new cycle as described by the algorithm then gives network~$N(X_1)$. Note that other valid subsets of the vertices of $\Omega^b(\dB|X_1)$ are~$\{a,b,d\}$, $\{a,b,c\}$, $\{a,b,c,d\}$ and $\{d\}$, which lead to alternative solutions.}
  \label{fig:binetsex2}
\end{figure}

\leooo{Finally, when $|X|\leq 2$ (in some recursive step), the problem can be solved trivially. In particular, when~$|X|=1$, the algorithm outputs a network consisting of a single vertex labelled by the only element of~$X$, which is the root as well as the leaf of the network.}

This completes the description of the algorithm for binets. Its correctness is shown in the following theorem.

\begin{theorem}\label{thm:binets}
Given a set~$\dB$ of binets on a set~$X$ of taxa, there exists a polynomial-time algorithm that decides if there exists a binary level-1 network on~$X$ that displays all binets in~$\dB$, and that constructs such a network if it exists.
\end{theorem}
\begin{proof}
We prove by induction on~$|X|$ that the algorithm described above produces a binary level-1 network on~$X$ displaying~$\dB$ if such a network exists. The induction basis for~$|X|\leq 2 $ is clearly true. Now let~$|X|\geq 3$, $\dB$ a set of binets on~$X$ and assume that there exists some binary level-1 network on~$X$ displaying~$\dB$. There are two cases.

\leoo{First assume that the graph $\cR^b(\dB)$ is disconnected and has connected components $C_1,\ldots,C_p$. Then the algorithm recursively computes networks $N|C_1,\ldots ,N|C_p$ displaying the sets $\dB|C_1,\ldots ,\dB|C_p$ respectively. Such networks exist by Observation~\ref{obs:subset} and can be found by the algorithm by induction. It follows that the network~$N$ which is constructed by the algorithm displays all binets in~$\dB$ of which both taxa are in the same connected component of $\cR^b(\dB)$. Each other binet is of the form~$T(x,y)$ by the definition of graph $\cR^b(\dB)$. Hence, those binets are also displayed by~$N$, by construction.}

\leoo{Now assume that the graph $\cR^b(\dB)$ is connected. Then we claim that there exists no binary level-1 network displaying~$\dB$ that is not cycle-rooted. To see this, assume that there exists such a network, let~$v_1,v_2$ be the two children of its root and~$X_i$ the leaves reachable by a directed path from~$v_i$, for~$i=1,2$. Then there is no edge~$\{a,b\}$ in~$\cR^b(\dB)$ for any $a\in X_1$ and~$b\in X_2$. Since $X_1\cup X_2=X$, it follows that~$\cR^b(\dB)$ is disconnected, which is a contradiction. Hence, any network that is a valid solution is cycle-rooted. The algorithm then searches for a nonempty strict subset~$U$ of the vertices of $\Omega^b(\dB)$ with no incoming arc, i.e., for which there is no arc~$(\pi_1,\pi_2)$ with $\pi_1\notin U$ and $\pi_2\in U$.}

\leoo{First assume that there exists no such set~$U$. Then the algorithm reports that there exists no solution. To prove that this is correct, assume that~$N'$ is some binary level-1 network on~$X$ displaying~$\dB$ and let~$H$ be the set of leaves that are high in~$N'$. The graph $\cK^b(\dB)$ contains no edges between taxa that are high in~$N'$ and taxa that are low in~$N'$ (because such taxa~$x,y$ cannot be together in a $T(x,y)$ binet). Hence, the set~$H$ is a union of vertex sets of connected components of~$\cK^b(\dB)$ and their representing vertices of $\Omega^b(\dB)$ form a subset~$U$. Moreover, there is no arc~$(\pi_1,\pi_2)$ in $\Omega^b(\dB)$ with $\pi_1\notin U$ and $\pi_2\in U$ because the corresponding binet $N(y;x)$ would not be displayed by~$N'$. Hence, we have obtained a contradiction to the assumption that there is no such set~$U$.}

\leoo{Now assume that there exists such a set~$U$. Then the algorithm recursively constructs networks $N(H)$ displaying $\dB|H$ and $N(\low)$ displaying $\dB|\low$, with~$H$ the union of the vertex sets of the connected components of $\cK^b(\dB)$ corresponding to the elements of~$U$, and with $\low=X\setminus H$. The algorithm then constructs a network~$N$ consisting of a cycle with networks~$N(H)$ hanging from the side of the cycle and network~$N(\low)$ hanging below the cycle, as in Figure~\ref{fig:binetsex2}. Networks $N(H)$ and $N(\low)$ exist by Observation~\ref{obs:subset} and can be found by the algorithm by induction. Because these networks display $\dB|H$ and $\dB|\low$ respectively, each binet from~$\dB$ that has both its leaves high or both its leaves low in~$N$ is displayed by~$N$. Each other binet is of the form $N(x;y)$ with~$x$ low and~$y$ high in~$N$, because otherwise there would exist an element in~$U$ which would have an incoming arc in~$\Omega^b(\dB)$. Hence, such binets are also displayed by~$N$.}
\end{proof}

\section{Constructing a network from a set of binets and trinets}\label{sec:nondense}

\subsection{Outline}\label{subsec:outline}

Let~$\cT$ be a set of binary level-1 trinets and binets on a set~$X$ of taxa. In this section we will describe an exponential-time algorithm for deciding whether there exists a binary level-1 network~$N$ on~$X$ with~$\cT\subseteq \cT(N)$ \cel{(Note that, if $\cT$ contains trinets or binets that are of a level greater than one, we know that such a network cannot exist)}.

Throughout this section, we will assume that there exists some binary level-1 network on~$X$ that displays~$\cT$ and we will show that in this case we can reconstruct such a network~$N$. 

Our approach aims at constructing the network~$N$ recursively; the recursive steps that are used depend on the structure of~$N$. The main steps of our approach are the following:

\begin{description}
\item[1] we determine whether the network~$N$ is cycle-rooted (see Section \ref{subsec:root});
\item[2] if this is the case, we guess the \emph{high} and \emph{low} leaves of~$N$ (see Section \ref{subsec:high});
\item[3] then, we guess how to partition the high leaves into the ``left'' and ``right'' leaves (see Section \ref{subsec:leftright});
\item[4] finally, we determine how to partition the leaves on each side into the leaves of the different sidenetworks on that side (see Section~\ref{subsec:sidenetworks}).
\end{description}

\leoo{Although we could do Steps 2 and 3 in a purely brute force way, we present several structural lemmas which restrict the search space and will be useful in Section~\ref{sec:tinycycle}.}

Once we have found a correct partition of the leaves (i.e., after Step~4), we recursively compute networks for each block of the partition and combine them into a single network. In the case that the network is not cycle-rooted, we do this by creating a root and adding arcs from this root to the recursively computed networks. Otherwise, the network is cycle-rooted. In this case, we construct a cycle with outgoing cut-arcs to the roots of the recursively computed networks, as illustrated in Figure~\ref{fig:network}.

The fact that we can recursively compute networks for each block of the computed partition follows from Observation~\ref{obs:subset}.

In the next sections we present a detailed description of our algorithm to reconstruct $N$. We will illustrate the procedure by applying it to the example set of trinets depicted in Figure~\ref{fig:exnets}.

\begin{figure}
  \centerline{\includegraphics{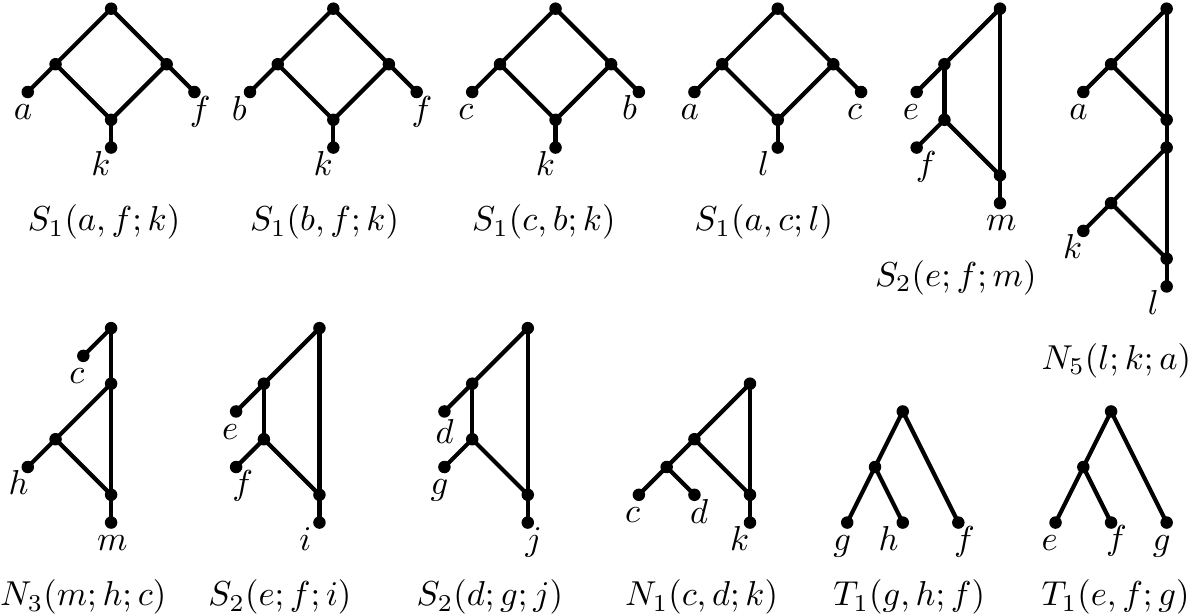}}
  \caption{The set~$\cT$ of trinets that we use to illustrate the inner workings of our algorithm.}
  \label{fig:exnets}
\end{figure}

\subsection{Is the network cycle-rooted?}\label{subsec:root}

To determine whether or not~$N$ is cycle-rooted, we define a graph~$R(\cT)$ as follows. The vertex set of~$R(\cT)$ is the set~$X$ of taxa and the edge set has an edge~$\{a,b\}$ if there exists a trinet or binet~$T\in\cT$ with~$a,b\in \cL(T)$ that is cycle-rooted or contains a common ancestor of~$a$ and~$b$ different from the root of~$T$ (or both). For an example, see Figure~\ref{fig:R}.

\begin{figure}
 \centering
  \includegraphics{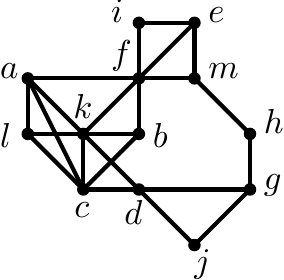}
 \caption{The graph~$\cR(\cT)$ for the set~$\cT$ of trinets in Figure~\ref{fig:exnets}. Since~$\cR(\cT)$ is connected, any network displaying~$\cT$ is cycle-rooted.}
 \label{fig:R}
\end{figure}

\begin{lemma}\label{lem:disconnected}
Let~$N$ be a binary level-1 network and~$\cT\subseteq\cT(N)$. If~$R(\cT)$ is disconnected and has connected components~$C_1,\ldots,C_p$, then~$\cT$ is displayed by the binary level-1 network~$N'$ obtained by creating a new root~$\rho$ and adding arcs from~$\rho$ to the roots of~$N|C_1,\ldots ,N|C_p$, and refining arbitrarily the root~$\rho$ in order to make the resulting network binary.
\end{lemma}
\begin{proof}
By Observation~\ref{obs:subset}, $N'$ displays each binet and each trinet of~$\cT$ whose leaves are all in the same connected component of~$R(\cT)$. Consider a binet~$B\in\cT$ on~$\{a,b\}$ with~$a$ and~$b$ in different components. Then there is no edge~$\{a,b\}$ in~$R(\cT)$ and hence~$B$ is not cycle-rooted, i.e.~$B=T(a,b)$, and~$B$ is clearly displayed by~$N'$. Now consider a trinet~$T\in\cT$ on~$\{a,b,c\}$ with~$a,b,c$ in three different components. Then, none of~$\{a,b\},\{b,c\}$ and~$\{a,c\}$ is an edge in~$R(\cT)$. Hence, none of the pairs~$\{a,b\},\{b,c\},\{a,c\}$ has a common ancestor other than the root of~$T$. Employing Figure~\ref{fig:nets}, this is impossible and so~$T$ cannot exist. Finally, consider a trinet~$T'\in\cT$ on~$\{a,b,c\}$ with~$a,b\in C_i$ and~$c\in C_j$ with~$i\neq j$. Then there is no edge~$\{a,c\}$ and no edge~$\{b,c\}$ in~$R(\cT)$. Consequently,~$T'$ is not cycle-rooted and the pairs~$\{a,c\}$ and~$\{b,c\}$ do not have a common ancestor in~$T'$ other than the root of~$T'$. Hence, $T'\in\{T_1(a,b;c) , N_3(a;b;c), N_3(b;a;c)\}$. If~$T'=T_1(a,b;c)$, then~$N|C_i$ displays~$T'|C_i=T(a,b)$ and hence~$N'$ displays~$T'$. If~$T'=N_3(a;b;c)$, then~$N|C_i$ displays~$T'|C_i=N(a;b)$ and, so,~$N'$ displays~$T'$. Symmetrically, if~$T'=N_3(b;a;c)$, then~$N|C_i$ displays~$T'|C_i=N(b;a)$ and, so,~$N'$ displays~$T'$. We conclude that~$N'$ displays~$\cT$.
\end{proof}

Hence, if~$R(\cT)$ is disconnected, we can recursively reconstruct a network for each of its connected components and combine the solutions to the subproblems in the way detailed in Lemma~\ref{lem:disconnected}.  

\cel{If $R(\cT)$ is connected, then we can apply the following lemma:}

\begin{lemma}\label{lem:cycle}
Let~$N$ be a binary level-1 network on~$X$ and~$\cT\subseteq\cT(N)$. If~$R(\cT)$ is connected and~$|X|\geq 2$, then~$N$ is cycle-rooted.
\end{lemma}
\begin{proof}
Suppose to the contrary that $R(\cT)$ is connected and that~$N$ is not cycle-rooted. Let~$v_1,v_2$ be the two children of the root of~$N$ and~$X_i$ the leaves of~$N$ reachable by a directed path from~$v_i$, for~$i=1,2$. \leooo{Note that $X_1\cap X_2=\emptyset$ and $X_1\cup X_2=X$.} Let~$a\in X_1$ and~$b\in X_2$ and let~$T$ be any trinet or binet displayed by~$N$ that contains~$a$ and~$b$. Then we have that~$T$ is not cycle-rooted and that the only common ancestor of~$a$ and~$b$ in~$T$ is the root of~$T$. Hence, there is no edge~$\{a,b\}$ in~$R(\cT)$ for any $a\in X_1$ and~$b\in X_2$, which implies that~$R(\cT)$ is disconnected; a contradiction.
\end{proof}

\cel{In the remainder of this section, we assume that~$R(\cT)$ is connected and thus that~$N$ is cycle-rooted.}

\subsection{Separating the high and the low leaves\label{subsec:high}} 

We define a graph~$\cK(\cT)$ whose purpose is to help decide which leaves are at the same height in~$N$. The vertex set of~$\cK(\cT)$ is the set of taxa~$X$ and the edge set contains an edge~$\{a,b\}$ if there exists a trinet or binet~$T\in\cT$ with~$a,b\in \cL(T)$ and in which~$a$ and~$b$ are at the same height in~$T$.

\begin{lemma}\label{lem:sameheight}
Let~$N$ be a cycle-rooted binary level-1 network and~$\cT\subseteq\cT(N)$. If~$C$ is a connected component of~$\cK(\cT)$, then all leaves in~$C$ are at the same height in~$N$.
\end{lemma}
\begin{proof}
We prove the lemma by showing that there is no edge in~$\cK(\cT)$ between any two leaves that are not at the same height in~$N$. Let~$h$ and~$\ell$ be leaves that are, respectively, high and low in~$N$. Then, in any trinet or binet~$T$ displayed by~$N$ that contains~$h$ and~$\ell$, we have that~$T$ is cycle-rooted and that~$h$ is high in~$T$ and~$\ell$ is low in~$T$. Hence, there is no edge~$\{h,\ell\}$ in~$\cK(\cT)$.
\end{proof}

We now define a directed graph~$\Omega(\cT)$ whose purpose is to help decide which leaves are high and which ones are low in~$N$. The vertex set of~$\Omega(\cT)$ is the set of connected components of~$\cK(\cT)$ and the arc set contains an arc~$(\pi,\pi')$ precisely if there is a cycle-rooted binet or trinet~$T\in\cT$ with~$h,\ell\in \cL(T)$ with~$h\in V(\pi)$ high in~$T$,~$\ell\in V(\pi')$ low in~$T$. See Figure~\ref{fig:K-Omega} for an example for both graphs.

\begin{figure}
 \centering
  \includegraphics{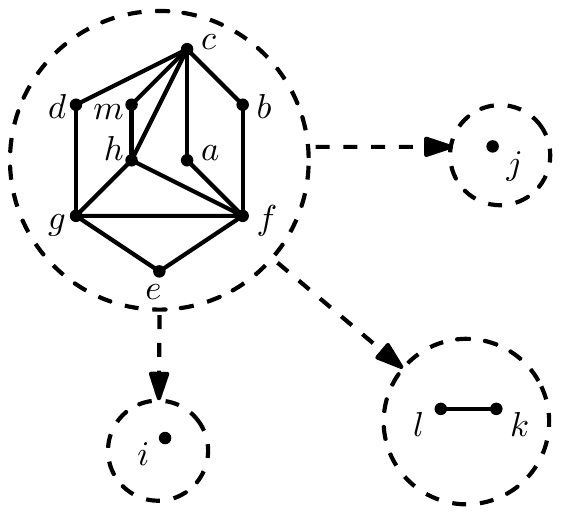}
 \caption{The graph~$\cK(\cT)$ in solid lines and the directed graph~$\Omega(\cT)$ in dashed lines, for the set~$\cT$ of trinets in Figure~\ref{fig:exnets}.}
 \label{fig:K-Omega}
\end{figure}

\begin{lemma}\label{lem:high}
Let~$\cT$ be a set of binets and trinets on a set~$X$ of taxa. Let~$N$ be a cycle-rooted binary level-1 network displaying~$\cT$. Then there exists a nonempty strict subset~$U$ of the vertices of~$\Omega(\cT)$ for which there is no arc~$(\pi,\pi')$ with~$\pi'\in U$, $\pi\notin U$ such that the set of leaves that are high in~$N$ equals $\cup_{\pi\in U} V(\pi)$.
\end{lemma}
\begin{proof}
Let~$H$ be the set of leaves that are high in~$N$. Note that~$H\neq\emptyset$. By Lemma~\ref{lem:sameheight},~$H$ is the union of connected components of~$\cK(\cT)$ and hence the union of a set~$U$ of vertices of~$\Omega(\cT)$ that respresent those components. We need to show that there is no arc $(\pi,\pi')$ with~$\pi'\in U$, $\pi\notin U$ in~$\Omega(\cT)$. To see this, notice that if there were such an arc, there would be a trinet or binet~$T\in\cT$ that is cycle-rooted and has leaves~$h,\ell\in \cL(T)$ with~$h\in V(\pi)$ high in~$T$ and~$\ell\in V(\pi')$ low in~$T$. However, such a trinet can only be displayed by~$N$ if either~$h$ is high in~$N$ and~$\ell$ is low in~$N$ or~$h$ and~$\ell$ are at the same height in~$N$. This leads to a contradiction because~$h\in V(\pi)$ with~$\pi\notin U$ and~$\ell\in V(\pi')$ with $\pi'\in U$ and, hence,~$h$ is low in~$N$ and~$\ell$ is high in~$N$.
\end{proof}

We now distinguish two cases. The first case is that the root of the network is in a cycle with size at least four, i.e., the network is largish-cycle rooted. The second case is that the root of the network is in a cycle with size three, i.e., that the network is tiny-cycle rooted. To construct a network from a given set of binets and trinets, the algorithm explores both options.

\subsubsection{The network is largish-cycle rooted\label{subsec:threearcs}} 

\leo{In this case,} we can simply try all subsets of vertices of~$\Omega(\cT)$ with no incoming arcs (i.e. arcs that begin outside and end inside the subset). For at least one such set~$U$ will hold that $\bigcup_{\pi\in U} V(\pi)$ is the set of leaves that are high in the network by Lemma~\ref{lem:high}. 

A set~$\cT$ of binets and trinets on a set~$X$ of taxa is called \emph{semi-dense} if for each pair of taxa from~$X$ there is at least one binet or trinet that contains both of them. If~$\cT$ is semi-dense, then we can identify the set of high leaves by the following lemma.

\begin{lemma}\label{lem:scc}
Let~$\cT$ be a semi-dense set of binets and trinets on a set~$X$ of taxa. Let~$N$ be a binary largish-cycle rooted level-1 network displaying~$\cT$. Let~$H$ be the set of leaves that are high in~$N$. Then \leooo{there is a the unique indegree-0 vertex~$\pi_0$ of~$\Omega(\cT)$ and~$H=V(\pi_0)$}.
\end{lemma}
\begin{proof}
\cel{
Since $\cT$ is semi-dense, for any two leaves~$h,h'\in H$ that are below different cut-arcs leaving the cycle~$C$ containing the root of~$N$, there exists a binet or a trinet $T$ in $\cT$ containing both $h$ an $h'$. Then, since $T$ is displayed by $N$, $T$ has to be a  binet or a trinet where  $h$ an $h'$ are at the same height. This implies that
there is an edge~$\{h,h'\}$ in~$\cK(\cT)$. Then, since there exist at least two different cut-arcs leaving~$C$,} the leaves in~$H$ are all in the same connected component of~$\cK(\cT)$. Then, by Lemma~\ref{lem:sameheight},~$H$ forms a connected component of~$\cK(\cT)$. Hence,~$H$ is a vertex of~$\Omega(\cT)$. This vertex has indegree-0 because no trinet or binet~$T$ displayed by~$N$ has a leaf~$\ell\notin H$ that is high in~$T$ and a leaf~$h\in H$ that is low in~$T$. Therefore,~$H$ is an indegree-0 vertex of~$\Omega(\cT)$. Moreover, by construction, there is an arc from~$H$ to each other vertex of~$\Omega(\cT)$. Hence,~$H$ is the unique indegree-0 vertex of~$\Omega(\cT)$.
\end{proof}

\subsubsection{The network is tiny-cycle rooted\label{subsec:twoarcs}} 
For this case, we define a modified graph~$\cK^\dagger(\cT)$, which is the graph obtained from~$\cK(\cT)$ as follows. For each pair of leaves~$a,b\in X$, we add an edge~$\{a,b\}$ if there is no such edge yet and there exists a largish-cycle rooted trinet~$T\in\cT$ with~$a,b\in \cL(T)$ (i.e.~$T$ is of type~$S_1(x,y;z)$ or~$S_2(x;y;z)$).

The directed graph~$\Omega^\dagger(\cT)$ is defined in a similar way as $\Omega(\cT)$ but its vertex set is the set of connected components of~$\cK^\dagger(\cT)$. Its arc set has, as in $\Omega(\cT)$, an arc~$(\pi,\pi')$ if there is a binet or trinet~$T\in\cT$ that is cycle-rooted and has leaves~$h,\ell\in \cL(T)$ with~$h\in V(\pi)$ high in~$T$,~$\ell\in V(\pi')$ low in~$T$.

\begin{figure}
 \centering
  \includegraphics{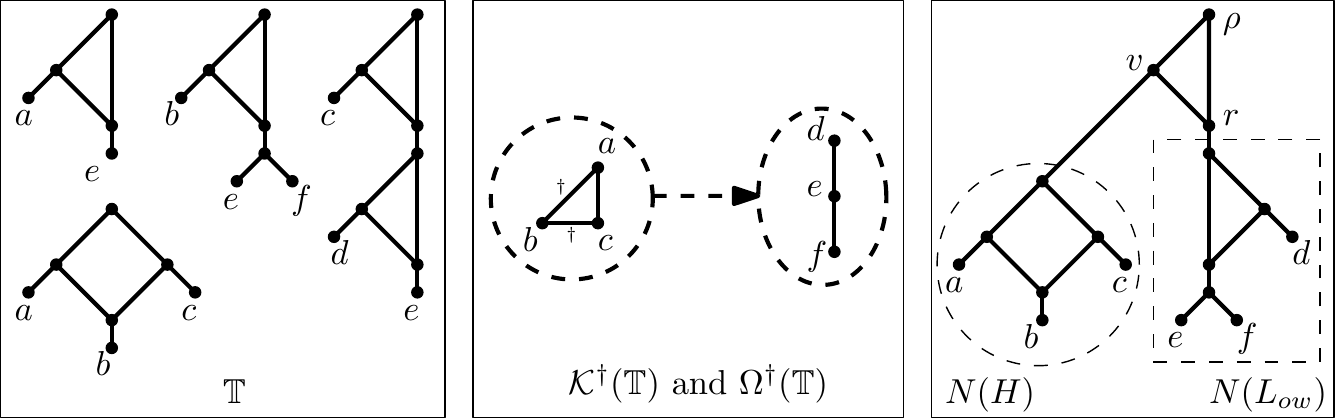}
 \caption{Example for the case that the network is tiny-cycle rooted. From left to right are depicted a set of trinets~$\cT$, its graphs $\cK^\dagger(\cT)$ (solid) and $\Omega^\dagger(\cT)$ (dashed) and the resulting network~$N$, obtained by combining networks~$N(H)$ and~$N(\low)$. Note that the two edges labelled~$\dagger$ are in $\cK^\dagger(\cT)$ but not in $\cK(\cT)$. }
 \label{fig:two-outgoing}
\end{figure}

Our approach for this case is to \cel{take} a non-empty strict subset~$U$ of the vertices of~$\Omega^\dagger(\cT)$ that has no incoming arcs and to take~$H$ to be the union of the elements of~$U$. Then, a network displaying~$\cT$ can be constructed by combining a network~$N(H)$ displaying~$\cT|H$ and a network~$N(\low)$ displaying~$\cT|\low$, with~$\low=X\setminus H$. (An example is depicted in Figure~\ref{fig:two-outgoing}). The next lemma shows how to construct such a network.

\begin{lemma}\label{lem:scc2}
Let~$\cT$ be a set of binets and trinets on a set~$X$ of taxa. Let~$N$ be a binary tiny-cycle rooted level-1 network displaying~$\cT$. Then there is a non-empty strict subset~$U$ of the vertices of~$\Omega^\dagger(\cT)$ such that there is no arc~$(\pi,\pi')$ with~$\pi'\in U$, $\pi\notin U$. Moreover, if~$U$ is any such set of vertices, then there exists a binary tiny-cycle rooted level-1 network~$N'$ displaying~$\cT$ in which~$\bigcup_{\pi\in U} V(\pi)$ is the set of leaves that are high in~$N'$.
\end{lemma}
\begin{proof}
Let~$H$ denote the set of leaves that are high in~$N$. Then, $H$ is the union of the vertex sets of one or more connected components of~$\cK(\cT)$ by Lemma~\ref{lem:sameheight}. Any largish-cycle rooted trinet which contains a leaf in~$H$ and a leaf not in~$H$ cannot be displayed by~$N$ because~$N$ is tiny-cycle rooted. Hence,~$H$ is also the union of one or more connected components of~$\cK^\dagger(\cT)$. These components form a subset~$U$ of the vertices of~$\Omega^\dagger(\cT)$. Furthermore, there is no arc~$(\pi,\pi')$ with~$\pi'\in U$, $\pi\notin U$ since no trinet or binet~$T$ displayed by~$N$ has a leaf that is in~$H$ and high in~$T$ and a leaf that is not in~$H$ and that is low in~$T$.

Now consider any nonempty strict subset~$U'$ of the vertices of~$\Omega^\dagger(\cT)$ with no incoming arcs, let~$H'$ be the union of the vertex sets of the corresponding connected components of $\cK^\dagger(\cT)$ and let $\low'=X\setminus H'$. Let~$N(H')$ be a binary level-1 network displaying~$\cT|H'$ and let~$N(\low')$ be a binary level-1 network displaying~$\cT|\low'$. Such networks exist by Observation~\ref{obs:subset}. Let~$N'$ be the network consisting of vertices~$\rho,v,r$, arcs~$(\rho,v)$, $(v,r)$, $(\rho,r)$, networks~$N(\low'), N(H')$ and an arc from~$v$ to the root of~$N(H')$ and an arc from~$r$ to the root of~$N(\low')$. Clearly, $N'$ is a tiny-cycle rooted level-1 network and~$H'$ is the set of leaves that are high in~$N'$.

It remains to prove that~$N'$ displays~$\cT$. First observe that for any~$h\in H'$ and~$\ell\in X\setminus H'$, there is no edge~$\{h,\ell\}$ in~$\cK^\dagger(\cT)$ (because otherwise~$h$ and~$\ell$ would lie in the same connected component). Hence, by construction of $\cK^\dagger(\cT)$, any binet or trinet containing~$h$ and~$\ell$ can not be tiny-cycle rooted and cannot have \cel{$h$ and~$\ell$ at the same height}. Moreover, in any such binet or trinet,~$h$ must be high and~$\ell$ must be low, because otherwise there would be an arc entering~$U'$ in~$\Omega^\dagger(\cT)$.

Consider any trinet or binet~$T\in\cT$. If the leaves of~$T$ are all in~$H'$ or all in~$L'$ then~$T\in\cT|H'$ or~$T\in\cT|L'$ and so~$T$ is clearly displayed by~$N'$. If~$T$ is a binet containing one leaf~$h\in H'$ and one leaf~$\ell\in X\setminus H'$, then~$T$ must be~$N(\ell;h)$ (by the previous paragraph) and, again,~$T$ is clearly displayed by~$N'$. Now suppose that~$T$ contains one leaf~$h\in H'$ and two leaves~$\ell,\ell'\in X\setminus H'$. \cel{Since we have argued in the previous paragraph that~$T$ is tiny-cycle rooted, $T$ must be of the form~$N_2(\ell,\ell';h)$, $N_5(\ell;\ell';h)$ or $N_5(\ell';\ell;h)$}. Moreover, since~$N(\low')$ displays the binet on~$\ell$ and~$\ell'$, and since~$h$ is high in~$T$ and~$\ell,\ell'$ low, it again follows that~$T$ is displayed by~$N'$. Finally, assume that~$T$ contains two leaves $h,h'\in H'$ and a single leaf~$\ell\in X\setminus H'$. Then (since~$T$ is tiny-cycle rooted)~$T$ must be of the form~$N_1(h,h';\ell)$ or $N_4(h;h';\ell)$. Since~$N(H')$ displays the binet on~$h$ and~$h'$, it follows that~$N'$ again displays~$T$.
\end{proof}

Note that the proof of Lemma~\ref{lem:scc2} describes how to build a tiny-cycle rooted level-1 network displaying~$\cT$ if such a network exists. Therefore, we assume from now on that the to be constructed network is largish-cycle rooted.

\subsection{Separating the left and the right leaves}\label{subsec:leftright} 

The next step is to divide the set~$H$ of leaves that are high in~$N$ into the leaves that are ``on the left'' and the leaves that are ``on the right'' of the cycle containing the root or, more precisely, to find the bipartition of~$H$ induced by some network displaying a given set of binets and trinets. We use the following definition.

\begin{definition}\label{def:feasible}
A bipartition of some set~$H\subseteq X$ is called \emph{feasible} with respect to a set of binets and trinets~$\cT$ if the following holds: 
\begin{enumerate}
\item[(F1)] if there is a binet or trinet~$T\in\cT$ containing leaves~$a,b\in H$ that has a common ancestor in~$T$ that is not the root of~$T$, then~$a$ and~$b$ are in the same part of the bipartition and
\item[(F2)] if there is a trinet~$S_1(x,y;z)\in\cT$ with~$x,y\in H$ and~$z\in X\setminus H$, then~$x$ and~$y$ are in different parts of the bipartition.
\end{enumerate}
\end{definition}

\leooo{Note that one part of a feasible bipartition may be empty.}

\begin{lemma}\label{lem:feasible}
Let~$N$ be a cycle-rooted binary level-1 network, let~$\cT\subseteq\cT(N)$, let~$H$ be the set of leaves that are high in~$N$ and let~$\{L,R\}$ be the bipartition of~$H$ induced by~$N$. Then $\{L,R\}$ is feasible with respect to~$\cT$.
\end{lemma}
\begin{proof}
First consider a binet or trinet~$T\in\cT$ containing leaves~$a,b\in H$ that have a common ancestor in~$T$ that is not the root of~$T$. Since~$N$ displays~$T$, it follows that~$a$ and~$b$ have a common ancestor in~$N$ that is not the root of~$N$. Hence,~$a$ and~$b$ are on the same side in~$N$. Since~$\{L,R\}$ is the bipartition of~$H$ induced by~$N$, it now follows that~$a$ and~$b$ are in the same part of the bipartition, as required.

Now consider a trinet~$S_1(x,y;z)\in\cT$ with~$x,y\in H$ and~$z\in X\setminus H$. Since~$N$ displays~$T$, we have that~$x$ and~$y$ are not on the same side in~$N$. Since~$\{L,R\}$ is the bipartition of~$H$ induced by~$N$, it follows that~$x$ and~$y$ are contained in different parts of the bipartition, as required.
\end{proof}

We now show how a feasible bipartition of a set~$H\subseteq X$ can be found in polynomial time. We define a graph~$\cM(\cT,H)=(H,E(\cM))$ with an edge~$\{a,b\}\in E(\cM)$ if there is a trinet or binet~$T\in\cT$ with~$a,b\in \cL(T)$ distinct and in which there is a common ancestor of~$a$ and~$b$ that is not the root of~$T$. The idea behind this graph is that leaves that are in the same connected component of this graph have to be in the same part of the bipartition.

Now define a graph~$\cW(\cT,H)=(V(\cW),E(\cW))$ as follows. The vertex set~$V(\cW)$ is the set of connected components of~$\cM(\cT,H)$ and there is an edge~$\{\pi,\pi'\}\in E(\cW)$ precisely if there exists a trinet~$S_1(x,y;z)\in\cT$ with~$x\in V(\pi)$, $y\in V(\pi')$ and~$z\in X\setminus H$. The purpose of this graph is to ensure that groups of leaves are in different parts of the bipartition, whenever this is necessary. See Figure~\ref{fig:M-W} for an example.

\begin{figure}
 \centering
  \includegraphics{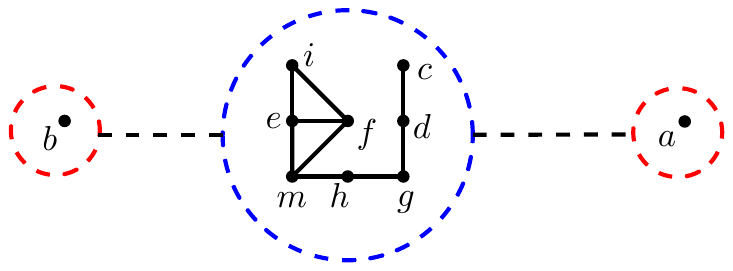}
 \caption{The graph~$\cM(\cT,H)$ in solid lines and the graph~$\cW(\cT,H)$ in dashed lines, for the set~$\cT$ of trinets in Figure~\ref{fig:exnets} and~$H=\{a,b,c,d,e,f,g,h,i,m\}$. A proper 2-colouring of~$\cW(\cT,H)$ is to color $\{a\}$ and $\{b\}$ in red and $\{c,d,e,f,g,h,i,m\}$ in blue.}
 \label{fig:M-W}
\end{figure}

\begin{lemma}\label{lem:2colouring}
Let~$\cT$ be a set of binets and trinets on~$X$ and~$H\subseteq X$. A bipartition~$\{L,R\}$ of~$H$ is feasible with respect to~$\cT$ if and only if
\begin{enumerate}
\item[(I)] $V(\pi)\subseteq L$ or~$V(\pi)\subseteq R$ for all~$\pi\in V(\cW)$ and
\item[(II)] there does not exist~$\{\pi,\pi'\}\in E(\cW)$ with~$(V(\pi)\cup V(\pi'))\subseteq L$ or~$(V(\pi)\cup V(\pi'))\subseteq R$.
\end{enumerate}
\end{lemma}
\begin{proof}
The lemma follows directly from observing that (F1) holds if and only if (I) holds and that (F2)  holds if and only if (II) holds.
\end{proof}

By Lemma~\ref{lem:2colouring}, all feasible bipartitions can be found by finding all 2-colourings of the graph~$\cW(\cT,H)$.

For example, consider the input set of trinets~$\cT$ from Figure~\ref{fig:exnets}. Since~$\cT$ is not semi-dense, we have to guess which connected components of~$\cK(\cT)$ form the set~$H$ of leaves that are high in the network (see Section~\ref{subsec:high}). If we guess~$H=\{a,b,c,d,e,f,g,h,i,m\}$, then we obtain the graphs $\cM(\cT,H)$ and~$\cW(\cT,H)$ as depicted in Figure~\ref{fig:M-W}. The only possible 2-colouring (up to symmetry) of the graph~$\cW(\cT,H)$ is indicated in the figure. From this we can conclude that~$a$ and~$b$ are on the same side of the network and that all other high leaves ($c,d,e,f,g,h,m$) are ``on the other side'' (i.e., none of them is on the same side as~$a$ or~$b$).

\subsection{Finding the pendant sidenetworks}\label{subsec:sidenetworks}

The next step is to divide the leaves of each part of the bipartition of the set of high leaves of the network into the leaves of the pendant sidenetworks. For this, we define the following graph and digraph.

Let~$\cT$ be a set of binets and trinets on~$X$, let $H\subsetneq X$, let $\{L,R\}$ be some bipartition of~$H$ that is feasible with respect to~$\cT$ and let~$S'\subseteq S\in\{L,R\}$. Consider the graph~$\cO(\cT,S',H)$ with vertex set~$S'$ and an edge~$\{a,b\}$ if
\begin{itemize}
\item there exists a trinet or binet~$T\in\cT|S'$ with~$a,b\in \cL(T)$ that has a cycle that contains the root or a common ancestor of~$a$ and~$b$ (or both) or;
\item there exists a trinet~$T\in\cT$ with~$\cL(T)=\{a,b,c\}$ with~$c\notin H$ and such that~$c$ is low in~$T$ and~$a$ and~$b$ are high in~$T$ and both in the same pendant sidenetwork of~$T$ or;
\item $T_1(a,b;c)\in\cT|S'$ for some~$c\in S'$.  
\end{itemize}

The directed graph~$\cD(\cT,S',H)$ (possibly having loops) has a vertex for each connected component of $\cO(\cT,S',H)$ and it has an arc~$(\pi_1,\pi_2)$ (possibly, $\pi_1=\pi_2$) precisely if there is a trinet in~$\cT$ of the form $S_2(x;y;z)$ with $x\in V(\pi_1)$, $y\in V(\pi_2)$ and $z\notin H$.

For example, Figure~\ref{fig:resnets} shows the set of trinets from Figure~\ref{fig:exnets} restricted to the set~$S=R=\{c,d,e,f,g,h,i,m\}$. The corresponding graphs $\cO(\cT,R,H)$ and $\cD(\cT,R,H)$, with $H=\{a,b,c,d,e,f,g,h,i,m\}$, are depicted in Figure~\ref{fig:O-D}. 

\begin{figure}[H]
  \centerline{\includegraphics{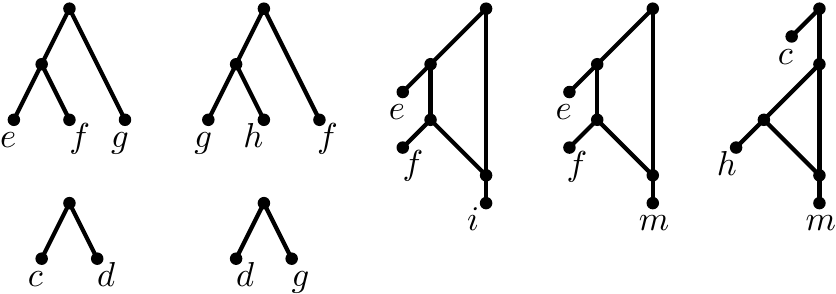}}
  \caption{The restricted set of trinets~$\cT|R$ with $R=\{c,d,e,f,g,h,i,m\}$ and~$\cT$ the set of trinets in Figure~\ref{fig:exnets} and $H=\{a,b,c,d,e,f,g,h,i,m\}$.}
  \label{fig:resnets}
\end{figure}

\begin{figure}[H]
 \centering
  \includegraphics{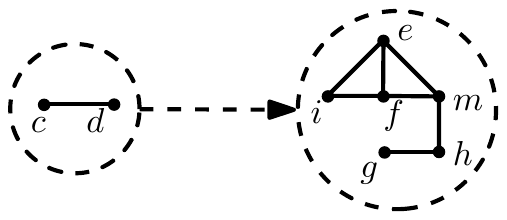}
 \caption{The graph~$\cO(\cT,R,H)$ in solid lines and the digraph~$\cD(\cT,R,H)$ in dashed lines, with $R$ and~$H$ as in Figure~\ref{fig:resnets}.}
 \label{fig:O-D}
\end{figure}

The following lemma shows that, if the digraph~$\cD(\cT,S',H)$ has no indegree-0 vertex, there exists no binary level-1 network displaying~$\cT$ in which~$H$ is the set of high leaves and all leaves in~$S'$ are on the same side.

\begin{lemma}\label{lem:noindeg0}
Let~$\cT$ be a set of binets and trinets on~$X$, let $H\subsetneq X$, let $\{L,R\}$ be a bipartition of~$H$ that is feasible with respect to~$\cT$ and let~$S'\subseteq S\in\{L,R\}$. If the graph $\cD(\cT,S',H)$ has no indegree-0 vertex, then there exists no binary level-1 network~$N$ that displays~$\cT$ in which~$H$ is the set of high leaves and all leaves in~$S'$ are on the same side.
\end{lemma}
\begin{proof}
Suppose that there exists such a network~$N$. Let~$\{L,R\}$ be the bipartition of~$H$ induced by~$N$ and suppose without loss of generality that~$L\cap S' \neq \emptyset$. Let~$L_1,\ldots ,L_q$ be the partition of~$L$ induced by the pendant sidenetworks of~$N$, \celB{ordered from the nearest to the farthest from the root}. Let~$i$ be the first index for which~$L_i\cap S'\neq \emptyset$. Then, by the definition of~$\cO(\cT,S',H)$,~$L_i\cap S'$ is the union of one or more connected components of~$\cO(\cT,S',H)$. Each of these connected components has indegree~0 in $\cD(\cT,S',H)$.
\end{proof}

Let~$\cT,X,H,L$ and~$R$ be as above. We present a \emph{sidenetwork partitioning algorithm}, which proceeds as follows for each~$S\in\{L,R\}$. Choose one indegree-0 vertex of $\cD(\cT,S,H)$ and call it~$S_1$. This will be the set of leaves of the first pendant sidenetwork on side~$S$. Then, construct the graph $\cO(\cT,S\setminus S_1,H)$ and digraph $\cD(\cT,S\setminus S_1,H)$, select an indegree-0 vertex and call it~$S_2$. Continue like this, i.e. let~$S_i$ be an indegree-0 vertex of $\cD(\cT,S\setminus (S_1\cup\ldots\cup S_{i-1} ),H)$, until an empty graph or a digraph with no indegree-0 vertex is obtained. In the latter case, there is no valid solution (under the given assumptions) by Lemma~\ref{lem:noindeg0}. Otherwise, we obtain sets $L_1,\ldots ,L_q$ and $R_1,\ldots ,R_{q'}$ containing the leaves of the pendant sidenetworks on both sides.

In the example in Figure~\ref{fig:O-D}, the only indegree-0 vertex of $\cD(\cT,R,H)$ is~$\{c,d\}$. Hence, we have~$R_1=\{c,d\}$. Since $\cO(\cT,R\setminus\{c,d\})$ is connected, $R_2=\{e,f,g,h,i,m\}$ follows.

We build a binary level-1 network~$N^*$ based on the sets $H,L_1,\ldots ,L_q, R_1,\ldots ,R_{q'}$ as follows. Let~$N(L_i)$ be a binary level-1 network displaying~$\cT|L_i$ for $i=1,\ldots ,q$ and let~$N(R_i)$ be a binary level-1 network displaying~$\cT|R_i$ for $i=1,\ldots ,q'$ \leooo{(note that it is possible that one of~$q$ and~$q'$ is~0.)}. We can build these networks recursively, and they exist by Observation~\ref{obs:subset}. In addition, we recursively build a network~$N(\low)$ displaying~$\cT|\low$ with~$\low=X\setminus H$. Now we combine these networks into a single network~$N^*$ as follows. We create a root~$\rho$, a reticulation~$r$, and two directed paths $(\rho,u_1,\ldots ,u_q,r)$, $(\rho,v_1,\ldots ,v_{q'},r)$ from~$\rho$ to~$r$  \leooo{(if~$q=0$ (respectively~$q'=0$) then there are no internal vertices on the first (resp. second) path)}. Then we add an arc from~$u_i$ to the root of~$N(L_i)$, for~$i=1,\ldots ,q$, we add an arc from~$v_i$ to the root of~$N(R_i)$ for~$i=1,\ldots ,q'$ and, finally, we add an arc from~$r$ to the root of~$N(\low)$. This completes the construction of~$N^*$. For an example, see Figure~\ref{fig:network}. 

We now prove that the network~$N^*$ constructed in this way displays the input trinets, assuming that 
there exists some solution that has~$H$ as its set of high leaves and~$\{L,R\}$ as the bipartition of~$H$ induced by it.

\begin{lemma}\label{lem:displays}
Let~$\cT$ be a set of binets and trinets, let~$N$ be a cycle-rooted binary level-1 network displaying~$\cT$, let~$H$ be the set of leaves that are high in~$N$ and let~$\{L,R\}$ be the feasible bipartition of~$H$ induced by~$N$. Then the binary level-1 network~$N^*$ constructed above displays~$\cT$.
\end{lemma}
\begin{proof}
The proof is by induction on the number~$|\cL(\cT)|$ of leaves in~$\cT$. The induction basis for $|\cL(\cT)|\leq 2$ is trivial. Hence, assume $|\cL(\cT)|\geq 3$.

For each pendant subnetwork~$N'$ of~$N^*$ with leaf-set~$X'$, there exists a binary level-1 network displaying~$\cT|X'$ by Observation~\ref{obs:subset}. Hence, the network~$N'$ that has been computed recursively by the algorithm displays~$\cT|X'$ by induction. It follows that any trinet or binet whose leaves are all in the same pendant sidenetwork of~$N^*$ is displayed by~$N^*$. Hence, it remains to consider binets and trinets containing leaves in at least two different pendant subnetworks of~$N^*$.

Let~$B\in\cT$ be a binet on leaves that are in two different pendant subnetworks of~$N^*$. If~$B=N(y;x)$ then, because~$B$ is displayed by~$N$,~$y$ is low in~$N$ and hence also low in~$N^*$. Since~$x$ and~$y$ are in different pendant subnetworks of~$N^*$, it follows that~$x$ is high in~$N^*$ and hence~$B$ is displayed by~$N^*$. If~$B=T(x,y)$ then there is an edge~$\{x,y\}$ in~$\cK(\cT)$ and hence~$x$ and~$y$ are at the same height in~$N^*$. Since~$x$ and~$y$ are in different pendant subnetworks, both must be high in~$N^*$ and it follows that~$N^*$ displays~$B$.

Now consider a trinet~$T\in\cT$ on leaves~$x,y,z$ that are in at least two different pendant subnetworks. At least one of~$x,y,z$ is high in~$N^*$ because otherwise all three leaves would be in the same pendant subnetwork~$N(\low)$, with~$\low=X\setminus H$. We now consider the different types of trinet that~$T$ can be.

First suppose that~$T=T_1(x,y;z)$. Then~$x,y,z$ form a clique in~$\cK(\cT)$ and hence all of~$x$, $y$ and~$z$ are high in~$N^*$. Moreover, by feasibility,~$x$ and~$y$ are in the same part of the bipartition~$\{L,R\}$ and hence on the same side in~$N^*$. If~$x$ and~$y$ are in the same pendant subnetwork then the binet $T|\{x,y\}=T(x,y)$ is displayed by this pendant subnetwork. Hence, in that case,~$T$ is clearly displayed by~$N^*$. Now assume that~$x$ and~$y$ are in different pendant subnetworks and assume without loss of generality that~$x,y\in R$. If~$z\in L$ then, again,~$T$ is clearly displayed by~$N^*$. Hence assume that~$x,y,z\in R$. Then, for each set~$R'\subseteq R$ containing~$x,y,z$, the graph $\cO(\cT,R',H)$ has an edge between~$x$ and~$y$. Hence, either~$x$ and~$y$ are in the same pedant sidenetwork, or~$z$ is in a pendant sidenetwork above the pendant sidenetworks that contain~$x$ and~$y$. Hence,~$T$ is displayed by~$N^*$.

Now suppose that~$T\in\{N_1(x,y;z),N_4(x;y;z)\}$. Then there is an edge~$\{x,y\}$ in~$\cK(\cT)$ and hence~$x$ and~$y$ are at the same height in~$N^*$. First note that~$x,y$ and~$z$ are not all high in~$N^*$ because otherwise~$x$,~$y$ and~$z$ would all be in the same part~$S$ of the bipartition~$\{L,R\}$ by feasibility and in the same pendant sidenetwork because they form a clique in~$\cO(\cT,S,H)$. Hence,~$z$ is not at the same height as~$x$ and~$y$ and hence~$z$ is not in the same connected component of~$\cK(\cT)$ as~$x$ and~$y$. Then there is an arc~$(\pi,\pi')$ in~$\Omega(\cT)$ with~$\pi$ the component containing~$x$ and~$y$ and~$\pi'$ the component containing~$z$. Hence~$x$ and~$y$ are high in~$N^*$ and~$z$ is low in~$N^*$ (since $\pi'$ has indegree greater than zero). Then,~$x$ and~$y$ are in the same part~$S$ of the bipartition~$\{L,R\}$ by feasibility and in the same pendant subnetwork of~$N^*$ because there is an edge~$\{x,y\}$ in~$\cO(\cT,S)$. Hence, since the binet $T|\{x,y\}$ is displayed by the pendant subnetwork containing~$x$ and~$y$, we conclude that~$T$ is displayed by~$N^*$.

Now suppose that~$T=S_1(x,y;z)$. We can argue in the same way as in the previous case that~$x$ and~$y$ are high in~$N^*$ and that~$z$ is low in~$N^*$. By feasibility,~$x$ and~$y$ are in different parts of the bipartition~$\{L,R\}$ and, hence,~$N^*$ displays~$T$.

Now suppose that~$T\in\{N_2(x,y;z),N_5(x;y;z)\}$. Then we can argue as before that~$x$ and~$y$ are at the same height in~$N^*$ and that~$z$ is not at the same height as~$x$ and~$y$ and hence that~$z$ is not in the same connected component of~$\cK(\cT)$ as~$x$ and~$y$. Then there is an arc~$(\pi,\pi')$ in~$\Omega(\cT)$ with~$\pi$ the component containing~$z$ and~$\pi'$ the component containing~$x$ and~$y$. Hence,~$z$ is high in~$N^*$ and~$x$ and~$y$ are low in~$N^*$. Since the binet $T|\{x,y\}$ is displayed by~$N(\low)$, we conclude that~$T$ is displayed by~$N^*$.

Now suppose that~$T=N_3(x;y;z)$. Observe that $x,y,z$ are all high in $N_3(x;y;z)$ because this trinet is not cycle-rooted. Therefore,~$x,y,z$ form a clique in~$\cK(\cT)$ and hence all of~$x$, $y$ and~$z$ are high in~$N^*$. Moreover, by feasibility,~$x$ and~$y$ are in the same part of the bipartition~$\{L,R\}$, say in~$R$, and hence on the same side in~$N^*$. First suppose that~$z\in L$. Then,~$\cT|R$ contains the binet $T|\{x,y\}$ which is cycle-rooted. Hence, there is an edge~$\{x,y\}$ in~$\cO(\cT,R,H)$ and~$x$ and~$y$ are in the same pendant sidenetwork. Since $T|\{x,y\}$ is displayed by this pendant subnetwork, it follows that~$T$ is displayed by~$N^*$. Now assume that~$z\in R$. Then the trinet~$T$ is in~$\cT|R$ and has a common ancestor of~$x$ and~$y$ contained in a cycle. Hence, as before,~$x$ and~$y$ are in the same pendant sidenetwork of~$N^*$ and, since $T|\{x,y\}$ is displayed by that pendant sidenetwork, it follows that~$T$ is displayed by~$N^*$.

Finally, suppose that~$T=S_2(x;y;z)$. As in the case $T\in\{N_1(x,y;z),N_4(x;y;z)\}$, we can argue that~$x$ and~$y$ are high in~$N^*$ and that~$z$ is low in~$N^*$. Then, by feasibility,~$x$ and~$y$ are on the same side~$S$ in~$N^*$. First suppose that~$x$ and~$y$ are in the same pendant sidenetwork of~$N^*$. Consider an iteration~$i$ of the sidenetwork partitioning algorithm  with $x,y\in S\setminus (S_1\cup\ldots\cup S_{i-1})$. Then there is an arc~$(\pi_1,\pi_2)$ in~$\cD(\cT,S\setminus (S_1\cup\ldots\cup S_{i-1}),H)$  with~$x\in V(\pi_1)$ and~$y\in V(\pi_2)$ (possibly $\pi_1=\pi_2$). Hence, $S_i$ does not contain~$y$ because~$\pi_2$ does not have indegree-0. It follows that~$x$ and~$y$ are in different sidenetworks and that the sidenetwork containing~$x$ is above the sidenetwork containing~$y$. Hence,~$N^*$ displays~$T$, which concludes the proof of the lemma.
\end{proof}

\begin{theorem}
There exists an $O(3^{|X|} poly(|X|))$ time algorithm that decides if there exists a binary level-1 network~$N$ that displays a given set~$\cT$ of binets and trinets.\label{thm:exptime}
\end{theorem}
\begin{proof}

If the graph~$R(\cT)$ is disconnected and has connected components~$C_1,\ldots,C_p$, then we recursively compute binary level-1 networks~$N_1,\ldots ,N_p$ displaying $\cT|C_1,\ldots ,\cT|C_p$ respectively. Then, by Lemma~\ref{lem:disconnected},~$\cT$ is displayed by the binary level-1 network~$N'$ obtained by creating a root~$\rho$ and adding arcs from~$\rho$ to the roots of~$N|C_1,\ldots ,N|C_p$, and refining the root~$\rho$ in order to make the network binary.

If~$R(\cT)$ is connected, then any binary level-1 network~$N$ displaying~$\cT$ is cycle-rooted by Lemma~\ref{lem:cycle}. If there exists such a network that is tiny-cycle rooted, then we can find such a network by Lemma~\ref{lem:scc2}. 

Otherwise, we can ``guess'', using Lemma~\ref{lem:high}, a set of leaves~$H$ such that there exists some binary level-1 network~$N$ displaying~$\cT$ in which~$H$ is the set of leaves that are high. Moreover, using Lemma~\ref{lem:2colouring}, we can ``guess'' a feasible partition~$\{L,R\}$ of~$H$ with respect to~$\cT$ by ``guessing'' a proper 2-colouring of the graph~$\cW(\cT,H)$. The total number of possible guesses for the tripartition~$\{L,R,X\setminus H\}$ is at most~$3^{|X|}$.

If there exists a binary level-1 network~$N'$ displaying~$\cT$ then, by Lemma~\ref{lem:displays}, there exists some tripartition ~$(L,R,X\setminus H)$ for which network~$N^*$ from Lemma~\ref{lem:displays} displays all binets and trinets in~$\cT$.

It remains to analyse the running time. Each recursive step takes $O(3^{|X|} poly(|X|))$ time and the number of recursive steps is certainly at most~$|X|$, leading to $O(3^{|X|} poly(|X|))$ in total since, by Observation~\ref{obs:subset}, the various recursive steps are independent of each other.
\end{proof}

Note that the running time analysis in the proof Theorem~\ref{thm:exptime} is pessimistic since, by Lemma~\ref{lem:sameheight}, the set~$H$ of high leaves must be the union of a subset of the vertices of~$\Omega(\cT)$ with no incoming arcs. Moreover, the number of feasible bipartitions of~$H$ is also restricted because each such bipartition must correspond to a 2-colouring of the graph $\cW(\cT,H)$. Hence, the number of possible guesses is restricted (but still exponential).

We conclude this section by extending Theorem~\ref{thm:exptime} to instances containing networks with arbitrarily many leaves.

\begin{corollary}\label{cor:exptime}
There exists an $O(3^{|X|} poly(|X|))$ time algorithm that decides if there exists a binary level-1 network~$N$ that displays a given set~$\cN$ of binary level-1 networks.\label{col:exptime}
\end{corollary}
\begin{proof}
We apply Theorem~\ref{thm:exptime} to the set~$\cT(\cN)$ of binets and trinets displayed by the networks in~$\cN$. To check that the resulting network~$N$ displays~$\cN$, consider a network~$N'\in\cN$. Since binary level-1 networks are encoded by their trinets~\cite{huber2011encoding}, any binary level-1 network displaying~$\cT(N')$ is equivalent to~$N'$. Hence,~$N|\cL(N')$ is equivalent to~$N'$. Therefore,~$N'$ is displayed  by~$N$.
\end{proof}

\section{Constructing a network from a set of tiny-cycle networks\label{sec:tinycycle}}

Recall that a network is a \emph{tiny-cycle} network if each cycle consists of exactly three vertices. It is easy to see that each tiny-cycle network is a level-1 network. Note that all binary level-1 binets and trinets except for $S_1(x,y;z)$ and $S_2(x;y;z)$ are tiny-cycle networks. We prove the following.

\begin{figure}[h]
 \centering
  \includegraphics{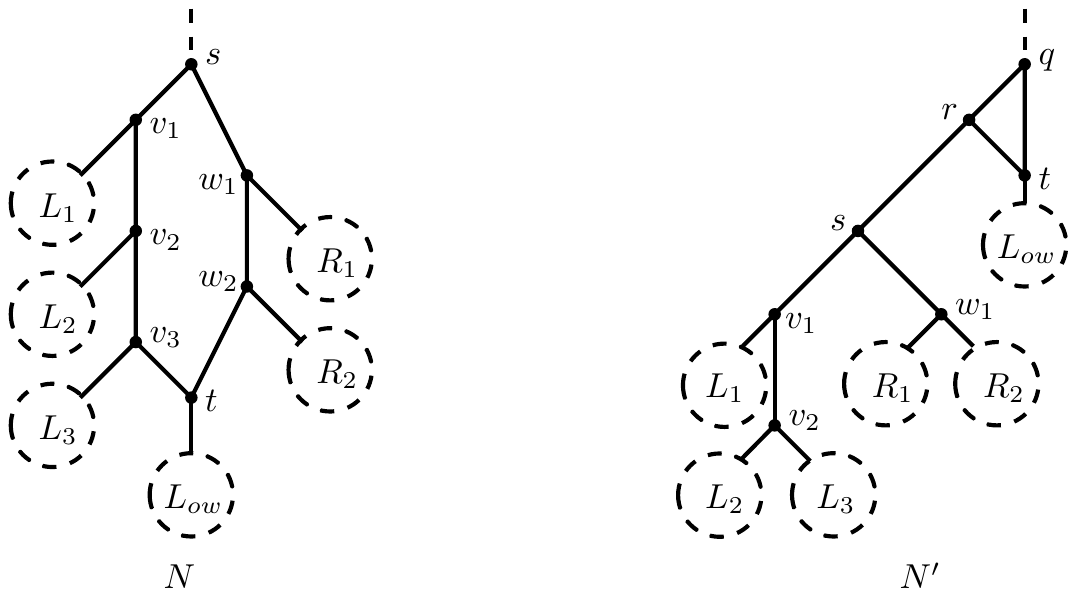}
 \caption{Transformation from a binary level-1 network~$N$ to a tiny-cycle network~$N'$, used in the proof of Theorem~\ref{thm:tinycycle} (with $n=3$ and $m=2$).}
 \label{fig:tinycycles}
\end{figure}

\begin{theorem}
Given a set~$\cT$ of tiny-cycle binets and trinets, we can decide in polynomial time if there exists a binary level-1 network displaying~$\cT$ and construct such a network if it exists.\label{thm:tinycycle}
\end{theorem}
\begin{proof}
Let~$N$ be a binary level-1 network displaying~$\cT$. If~$N$ is not a tiny-cycle network, then we construct a tiny-cycle network~$N'$ from~$N$ as follows (see Figure~\ref{fig:tinycycles} for an illustration). For each cycle of~$N$ consisting of internally vertex-disjoint directed paths $(s,v_1,\ldots ,v_n,t)$ and $(s,w_1,\ldots ,w_m,t)$ with $n+m\geq 2$, do the following. 
\leoo{Delete arcs $(v_n,t)$ and $(w_m,t)$, suppress $v_n$ and $w_m$, add vertices~$q$ and~$r$ and arcs $(q,r)$, $(r,t)$, $(q,t)$ and $(r,s)$. Finally, if~$s$ is not the root of~$N$, let~$p$ be the parent of~$s$ in~$N$ and replace arc $(p,s)$ by an arc~$(p,q)$.} Let~$N'$ be the obtained network. It is easy to verify that any binary tiny-cycle network that is displayed by~$N$ is also displayed by~$N'$ and that~$N'$ is a tiny-cycle network. Hence, we may restrict our attention to constructing tiny-cycle networks.

The only two cases to consider are that the to be constructed network is not cycle-rooted and that it is tiny-cycle rooted. \leoo{We can deal with the first case by Lemmas~\ref{lem:disconnected} and~\ref{lem:cycle} and with the second case by Lemma~\ref{lem:scc2}.}
\end{proof}

Note that Theorem~\ref{thm:tinycycle} applies to sets of binets and trinets that do not contain any trinets of the form $S_1(x,y;z)$ and $S_2(x;y;z)$. The following corollary generalises this theorem to general instances of tiny-cycle networks. Its proof is analogous to the proof of Corollary~\ref{cor:exptime}.

\begin{corollary}
Given a set~$\cN$ of tiny-cycle networks, we can decide in polynomial time if there exists a binary level-1 network displaying~$\cN$ and construct such a network if it exists.\label{cor:tinycycle}
\end{corollary}

\section{Complexity of constructing a network from a set of trinets\label{sec:hardness}}

In this section, we show that it is NP-hard to construct a binary level-1 network from a nondense set of trinets.

\begin{theorem}
\label{nphard}
Given a set of trinets~$\cT$, it is NP-hard to decide if there exists a binary level-1 network~$N$ displaying~$\cT$. In addition, it is NP-hard to decide if there exists such a network with a single reticulation.
\end{theorem}
\begin{proof}
We reduce from the NP-hard problem \textsc{SetSplitting} \cite{GareyJohnson1979}. The reduction is a nontrivial adaptation of the reduction given by Jansson, Nguyen and Sung for the triplet-version of the problem~\cite{JanssonEtAl2006}. 

The \textsc{SetSplitting} problem is defined as follows. Given a set~$U$ and a collection~$\cC$ of size-3 subsets of~$U$, decide if there exists a bipartition of~$U$ into sets~$A$ and~$B$ such that for each~$C\in\cC$ holds that $C\cap A\neq\emptyset$ and $C\cap B\neq\emptyset$? If such a bipartition exists, then we call it a \emph{set splitting} for~$\cC$.

The reduction is as follows. For an example see Figure~\ref{fig:reduc}. Assume that~$\cC=\{C_1,\ldots ,C_k\}$, with~$k\geq 1$, and furthermore that the elements of~$U$ are totally ordered (by an operation $<$). We create a taxon set~$X$ and a trinet set~$\cT$ on~$X$ as follows. For each~$u\in U$, put a taxon~$u_0$ in~$X$. In addition, add a special taxon~$b$ to~$X$. Then, for each~$C_i=\{u,v,w\}$ with~$u<v<w$, add taxa~$u_i,u_i',v_i,v_i',w_i,w_i'$ to~$X$ and add the following trinets to~$\cT$: $T_1(v_i,v_i';u_i)$, $T_1(w_i,w_i';v_i)$, $T_1(u_i,u_i';w_i)$, $S_2(v_i;v_i';b)$, $S_2(w_i;w_i';b)$, $S_2(u_i;u_i';b)$, $S_1(u_i,u_0;b)$, $S_1(v_i,v_0;b)$, $S_1(w_i,w_0;b)$, resulting in~$\cT$ comprising of three different types of trinets. This completes the reduction.

\begin{figure}
  \centerline{\includegraphics{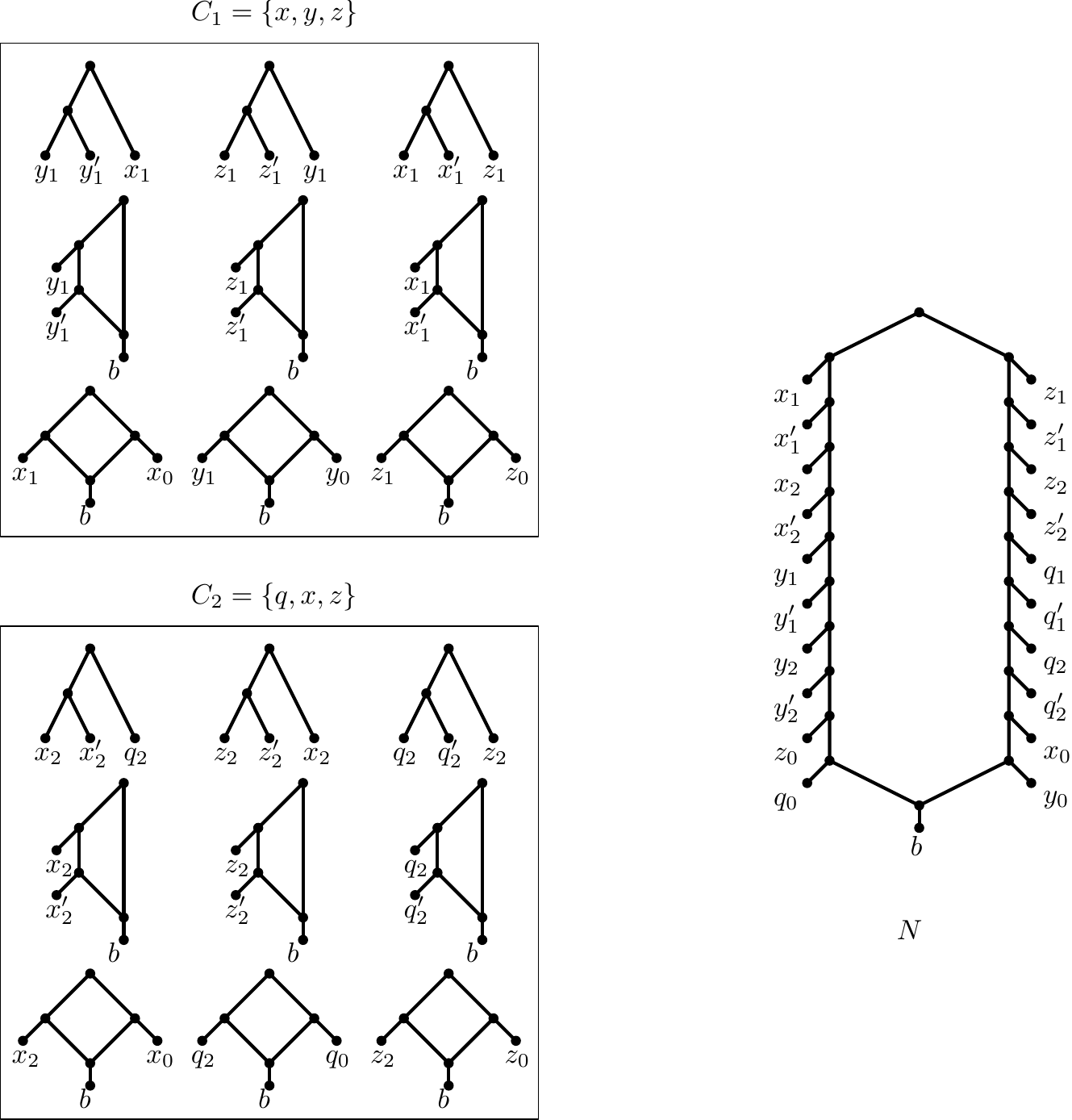}}
  \caption{An example of the reduction if the input of the \textsc{SetSplitting} problem is $\cC=\{\{x,y,z\},\{q,x,z\}\}$ (with $q<x<y<z$). A valid set splitting is~$\{\{q,z\}, \{x,y\}\}$ and the network~$N$ indicated in the figure displays all trinets produced by the reduction.}
  \label{fig:reduc}
\end{figure}

First we show that, if there exists a set splitting for~$\cC$, then there exists a network~$N$ on~$X$ that displays~$\cT$ and has exactly one reticulation. Let~$\{A,B\}$ be a set splitting for~$\cC$. We construct a network~$N$ whose root is in a cycle and each arc leaving this cycle ends in a leaf. Hence,~$N$ has precisely one reticulation, whose child is leaf. Label this leaf by taxon~$b$. Let~$\{L,R\}$ be the bipartition of~$X\setminus\{b\}$ induced by~$N$. Then, for each~$u\in A$ we put~$u_0$ in~$L$ and~$u_i$ and~$u_i'$ for~$i\geq 1$ in~$R$. \leoo{Symmetrically, for each~$u\in B$ we put~$u_0$ in~$R$ and~$u_i$ and~$u_i'$ for~$i\geq 1$ in~$L$.} It remains to describe the order of the leaves on each side of the network. For each~$i\in\{1,\ldots ,k\}$ and for any two leaves~$x_i,y_i$ that are on the same side in~$N$, put $x_i$ above $y_i$ and $y_i'$ if $x<y$ (and put $y_i$ above $x_i$ and~$x_i'$ if~$y>x$). In addition, put each~$x_i$ above~$x_i'$. The ordering can be completed arbitrarily. For an example, see the network in Figure~\ref{fig:reduc}. To see that~$N$ displays~$\cT$, first observe that all trinets of the form $S_2(x_i;x_i';b)$ are displayed by~$N$ because we put~$x_i$ above~$x_i'$ and on the same side. In addition, all trinets of the form $S_1(x_i,x_0;b)$ are also displayed by~$N$ because~$x_i$ and~$x_0$ are on opposite sides of~$N$. Now consider a constraint~$C_i=\{u,v,w\}\in\cC$ with~$u<v<w$. Since~$\{A,B\}$ is a set splitting, $|C_i\cap A|=2$ or $|C_i\cap B|=2$. Suppose that~$u$ and~$v$ are in the same set, say~$u,v\in A$. The other two cases can be dealt with in a similar manner. It then follows that~$w\in B$ and hence that $u_i,v_i,u_i',v_i'\in R$ and $w_i,w_i'\in L$. It then follows that trinets $T_1(w_i,w_i';v_i)$ and $T_1(u_i,u_i';w_i)$ are displayed by~$N$. Moreover, since $u<v$, we have that~$u_i$ is above~$v_i$ and~$v_i'$. Hence, also trinet $T_1(v_i,v_i';u_i)$ is displayed by~$N$. We conclude that~$\cT$ is displayed by~$N$.

It remains to show that if there exists a network on~$X$ that displays~$\cT$, then there exists a set splitting~$A,B$ for~$\cC$. So assume that there exists a network~$N$ on~$X$ that displays~$\cT$. From Lemma~\ref{lem:cycle} it follows that~$N$ is cycle-rooted. By Lemma~\ref{lem:sameheight}, all leaves in~$X\setminus\{b\}$ are at the same height in~$N$. Consequently, by Lemma~\ref{lem:high}, the leaves in~$X\setminus\{b\}$ are all high in~$N$ and~$b$ is low in~$N$. Let~$\{L,R\}$ be the bipartition of~$X\setminus\{b\}$ induced by~$N$. Define $A=\{u\in U \mid u_0 \in L\}$ and $B = \{u\in U \mid u_0\in R\}$. We claim that~$A$ and~$B$ form a set splitting for~$\cC$. To show this, assume the contrary, i.e., that there exists some constraint $C_i=\{u,v,w\}\in\cC$, with $u<v<w$, such that either $u,v,w\in A $ or $u,v,w\in B$. Assume without loss of generality that $u,v,w\in A$. Then $u_0,v_0,w_0\in L$. Hence, $u_i,v_i,w_i\in R$ since the trinets $S_1(u_i,u_0;b)$, $S_1(v_i,v_0;b)$ and $S_1(w_i,w_0;b)$ are contained in~$\cT$. Then, since $S_2(u_i;u_i';b)\in\cT$, we obtain that $u_i'\in R$. Moreover,~$u_i$ and~$u_i'$ are in different sidenetworks of~$N$. Then, since $T_1(u_i,u_i';w_i)\in\cT$, the sidenetwork containing~$w_i$ is strictly above the sidenetwork containing $u_i$. However, it follows in the same way from trinets $S_2(v_i;v_i';b)$ and $T_1(v_i,v_i';u_i)$ that the sidenetwork of~$N$ containing~$u_i$ is strictly above the sidenetwork containing~$v_i$. Furthermore, it follows from the trinets $S_2(w_i;w_i';b)$ and $T_1(w_i,w_i';v_i)$ that the sidenetwork containing~$v_i$ is strictly above the sidenetwork containing~$w_i$. Hence, we have obtained a contradiction. It follows that~$A$ and~$B$ form a set splitting of~$\cC$, completing the proof.
\end{proof}

\section{Concluding remarks\label{sec:concl}}

We have presented an exponential time algorithm for
determining whether or not an arbitrary set of binets and trinets
is displayed by a level-1 network, shown that this
problem is NP-hard, and
given some polynomial time algorithms for solving it
in certain special instances. It would be
interesting to know whether other special instances are
also solvable in polynomial time (for example, when either $S_1(x,y;z)$-type or $S_2(x;y;z)$-type trinets are excluded).

\leooo{We note that the problem of deciding if a set of binets and trinets is displayed by a level-1 network remains NP-hard when the input set is semi-dense, i.e. for each combination of two taxa it contains at least one binet or trinet containing those two taxa. Although the trinet set produced in the proof of Theorem~\ref{nphard} is not semi-dense, it is not difficult to make it semi-dense without affecting the reduction.}

As mentioned in the introduction, our algorithms can
be regarded as a supernetwork approach for constructing
phylogenetic networks. It is therefore worth noting
that our main algorithms extend to the case where the input consists of
a collection of binary level-1 networks,
where each input network is allowed to have any number of leaves.

Furthermore, we have shown that constructing a binary level-1 network from a set of trinets is NP-hard in general. One could instead consider the problem of constructing
such networks from networks on (at least) $m$ leaves, $3 \le m \le |X|$
(or $m$-nets for short). \leooo{However, it can be shown that it is also NP-hard}
to decide if a set of level-1 $m$-nets is displayed by a
level-1 network or not. \leooo{This can be shown by a simple reduction from the problem for trinets.}

It would also be of interest to develop algorithms to reconstruct level-$k$ networks for $k\ge 2$ from $m$-nets. Note that
the trinets displayed by a level-2 network always encode the network \cite{level2trinets}, but that in general trinets do not encode level-$k$ networks \cite{information}. Therefore, in light also of the results in this paper, we
anticipate that these problems
might be quite challenging in general. Even so, it could
still be very useful to develop heuristics for tackling
such problems as this has proven very useful in
both supertree and phylogenetic network construction.

\section*{Acknowledgements}
Leo van Iersel was partially supported by a Veni grant of the Netherlands Organization for Scientific Research (NWO). Katharina Huber and Celine Scornavacca thank the London Mathematical Society for partial support in the context of their computer science small grant scheme.
{This publication is the contribution no. 2014-XXX of the Institut des Sciences de l'Evolution de Montpellier (ISE-M, UMR 5554). This work has been partially funded by the French {\it Agence Nationale de la Recherche}, {\it Investissements d'avenir/Bioinformatique} (ANR-10-BINF-01-02, {\it Ancestrome}).

\bibliographystyle{amsplain}
\providecommand{\bysame}{\leavevmode\hbox to3em{\hrulefill}\thinspace}
\providecommand{\MR}{\relax\ifhmode\unskip\space\fi MR }
\providecommand{\MRhref}[2]{%
  \href{http://www.ams.org/mathscinet-getitem?mr=#1}{#2}
}
\providecommand{\href}[2]{#2}

\end{document}